\documentclass[runningheads,a4paper]{llncs}

\usepackage{
amsmath,
amscd,
amssymb,
}

\usepackage{algorithm}
\usepackage[noend]{algpseudocode}
\usepackage{hyperref}
\usepackage{comment}
\usepackage{tikz}
\usetikzlibrary{positioning,arrows,calc}

\usepackage{graphicx}

\usepackage{url} 
\newcommand{\keywords}[1]{\par\addvspace\baselineskip
\noindent\keywordname\enspace\ignorespaces#1}

\newcommand{\Z}{\mathbb{Z}}

\newcommand{\N}{\mathbb{N}}

\newcommand{\PI}{\mathrm{\Pi}}

\newcommand{\SFT}{\mathrm{SFT}}
\newcommand{\sofic}{\mathrm{sofic}}

\title{Plane-Walking Automata}
\titlerunning{Plane-Walking Automata}

\author{
Ville Salo
\and
Ilkka T\"orm\"a
}

\institute{
		TUCS -- Turku Centre for Computer Science \\
		University of Turku, Finland \\
		\email{\{vosalo,iatorm\}@utu.fi}
}

\begin{document}
\maketitle

\begin{abstract}
In this article, we study classes of multidimensional subshifts defined by multihead finite automata, in particular the hierarchy of classes of subshifts defined as the number of heads grows. The hierarchy collapses on the third level, where all co-recursively enumerable subshifts are obtained in every dimension. We also compare these classes to SFTs and sofic shifts. We are unable to separate the second and third level of the hierarchy in one and two dimensions, and suggest a related open problem for two-counter machines.
\end{abstract}


\section{Introduction}

In this article, we discuss multihead finite automata on infinite multidimensional configurations, which we call plane-walking automata, and use them to define classes of subshifts. Our model is based on the general idea of a graph-walking automaton. In this model, the automaton is placed on one of the nodes of a graph with colored nodes, and it repeatedly reads the color of the current node, updates its internal state, and steps to an adjacent node. The automaton eventually enters an accepting or rejecting state, or runs forever without making a decision. Usually, we collect the graphs that it accepts, or the ones that it does not reject, and call this collection the language of the automaton. We restrict our attention to machines that are deterministic, although an interesting continuation of our research would be to consider nondeterministic or alternating machines.

Well-known such models include the two-way deterministic finite automata (2DFA) walking back-and-forth on a finite word, and tree-walking automata traversing a tree. See \cite{HoKuMa09} for a survey on multihead automata on words, and the references in \cite{Bo08a} for information on tree-walking automata. In multiple dimensions, our automata are based on the concept of picture-walking (or $4$-way) automata for accepting picture languages, defined in \cite{BlHe67} and surveyed in \cite{InTa91,KaSa11}.

The first question about subshifts accepted by plane-walking automata is how this class relates to existing classes of subshifts. In particular, we compare the class of subshifts accepted by a one-head deterministic automaton to SFTs and sofic shifts, two well-known classes in the theory of subshifts. They correspond, in some sense, to local languages and regular languages of finite words, since an SFT is defined by local rules, and a sofic shift is a letter-to-letter projection of an SFT. It is well-known that in the one-dimensional finite case, graph-walking automata with a single head (2DFA) define precisely the regular languages. However, for more complicated graphs, deterministic graph-walking automata often define a smaller class than the one containing letter-to-letter projections of local languages (which is often considered the natural generalization of regularity): deterministic tree-walking automata do not define all regular tree languages \cite{BoCo08} and deterministic picture-walking automata do not accept all recognizable picture languages \cite{GiVeRe97}. We show in Theorem~\ref{thm:Between} that this is also the case for a one-head deterministic plane-walking automaton in the multidimensional case: the class of subshifts defined is strictly between SFTs and sofic shifts.

Already in \cite{BlHe67}, the basic model of picture-walking automata was augmented by multiple heads,\footnote{Strictly speaking, they were augmented by markers, but the difference is small.} and we similarly consider classes of subshifts defined by multihead plane-walking automata. In \cite[Theorem 3]{BlHe67}, it was shown that the hierarchy obtained as the number of heads grows is infinite in the case of pictures (by a diagonalization argument). Similar results are known for one-dimensional words \cite{HsYe75} and trees \cite{BuPrSaWe06}. In the case of subshifts, we show that the hierarchy collapses to the third level, which is precisely the class of subshifts whose languages are co-recursively enumerable. In particular, it properly contains the class of sofic shifts. However, we are not able to separate the second and third levels in the case of one or two dimensions, although we find it very likely that they are distinct. We discuss why this problem appears hard to us, suggest a possible separating language, and state a related open problem for two-counter machines.

\section{Preliminary Notions}

In this article, a \emph{($d$-dimensional) pattern} is a function $P : D \to \Sigma$, where $D = D(P) \subset \Z^d$ is the \emph{domain} of $P$, and $\Sigma$ is a finite \emph{alphabet}. A full pattern with domain $\Z^d$ is called a \emph{configuration (over $\Sigma$)}, and other patterns have finite domains unless otherwise noted. The restriction of a pattern $P$ to a smaller domain $D$ is denoted by $P|_D$. We say that a pattern $P$ \emph{occurs at $\vec v \in \Z^d$} in another pattern $P'$, if we have $\vec u + \vec v \in D(P')$ and $P'_{\vec u + \vec v} = P_{\vec u}$ for all $\vec u \in D(P)$. For $s \in \Sigma$, we denote by $|P|_s$ the number of occurrences of $s$ in $P$.

A \emph{subshift} is a set $X \subset \Sigma^{\Z^d}$ of configurations defined by a set $F$ of \emph{forbidden patterns} -- a configuration $x \in \Sigma^{\Z^d}$ is in $X$ if and only if none of the patterns of $F$ occur in it. If $F$ is finite, then $X$ is a \emph{subshift of finite type}, or SFT for short, and if $F$ is recursively enumerable, then $X$ is $\emph{co-RE}$ or $\PI^0_1$. If the domain of every pattern in $F$ is of the form is $\{\vec 0, \vec e_i\}$, where $\vec e_1, \ldots, \vec e_d$ is the natural basis of $\Z^d$, then $X$ is a \emph{tiling system}. A \emph{sofic shift} is obtained by renaming the symbols of an SFT, or equivalently a tiling system. If it is decidable whether a given pattern occurs in some configuration of $X$, then $X$ is \emph{recursive}.

Unless otherwise noted, we always use the binary alphabet $\Sigma = \{0,1\}$.

\section{Choosing the Machines}

The basic idea in this article is to define subshifts by deterministic and multihead finite automata as follows: Given a configuration $x \in \Sigma^{\Z^d}$, we initialize the heads of the automaton on some of its cells, and let them run indefinitely, moving around and reading the contents of $x$. If the automaton halts in a rejecting state, then we consider $x$ to be rejected, and otherwise it is accepted.

After this high-level idea has been established, there are multiple a priori inequivalent ways of formalizing it, and we begin with a discussion of such choices. Much of this freedom is due to the fact that many different definitions and variants of multihead finite automata exist in the literature, both in the case of finite or infinite pictures and one-dimensional words (see \cite{HoKuMa09} and references therein).

\emph{Heads or markers?} A multihead automaton can be defined as having multiple heads capable of moving around the input, or as having one mobile head and several immobile markers that the head can move around. In the latter case, one must also decide whether the markers are indistinguishable or distinct, and whether they can store information or not. In this article, we choose the former approach of having multiple mobile heads.

\emph{Global control or independent heads?} Next, we must choose how the heads of our machines interact. The traditional approach is to have a single global state that controls each head, but in our model, this could be considered `physically infeasible', as the heads may travel arbitrarily far from each other. For this reason, and in order not to have too strong a model, the heads of our automata are independent, and can interact only when they lie in the same cell.

\emph{Synchronous or asynchronous motion?} Now that the heads have no common memory, we need to decide whether they still have a common perception of time, that is, whether they can synchronize their motion. In the synchronous updating scheme, the heads update their states and positions at the same time, so that the distance between two heads moving in the same direction stays constant. The other option is asynchronous updating, where the heads may update at different paces, possibly nondeterministically. We choose the synchronous scheme, as it is easier to formalize and enables us to shoot carefully synchronized signals, which we feel are the most interesting aspect of multihead plane-walking automata.

Next, we need to decide how exactly a plane-walking automaton defines a subshift. Recall that a subshift is defined by a possibly infinite set of finite forbidden patterns in a translation-invariant way. In our model, the forbidden patterns should be exactly those that support a rejecting run of the automaton.

\emph{How do we start?} First, we could always initialize our automata at the origin $\vec 0 \in \Z^d$, decide the acceptance of a configuration based on this single run, and restrict to automata that define translation-invariant sets. Second, we may quantify over all coordinates of $\Z^d$, initialize all the heads at the same coordinate, and reject if some choice leads to rejection. In the third option, we quantify over all $k$-tuples of coordinates, and place the $k$ heads in them independently. The first definition is not very satisfying, since most one-head automata would have to be discarded, and of the remaining two, we choose the former, as it is more restrictive. We also quantify over a set of initial states, so that our subshift classes are closed under finite intersection, and accordingly seem more natural.

\emph{How do we end?} Finally, we have a choice of what constitutes as a rejecting state. Can a single head cause the whole computation to reject, or does every head have to reject at the same time, and if that is the case, are they further required to be at the same position? We again choose the most restrictive option.

All of the above models are similar, in that by adding a few more heads or counters, one can usually simulate an alternative definition. Sometimes, one can even show that two models are equivalent. For example, \cite[Theorem~2.3]{BlHe67} states that being able to distinguish markers is not useful in the case of finite pictures; however, the argument seems impossible to apply to plane-walking automata.

To recap, our definition of choice is the \emph{deterministic $k$-head plane-walking finite automaton with local information sharing, synchronous updating, quantification over single initial coordinate and initial state, and rejection with all heads at a single coordinate}, with the (necessarily ambiguous) shorthand $k$PWDFA.

\section{Definitions}

We now formally define our machines, runs, acceptance conditions and the subshifts they define. For this section, let the dimension $d$ be fixed.

\begin{definition}
A $k$PWDFA is a $5$-tuple $A = (Q, \Sigma, \delta, I, R)$, where $Q = Q_1 \times \cdots \times Q_k$ is the finite set of \emph{global states}, the $Q_i$ are the \emph{local states}, $\Sigma$ is the \emph{alphabet}, and $\delta = (\delta_1, \ldots, \delta_k)$ is the list of \emph{transition functions}
\[ \delta_j : S_j \times \Sigma \to Q_j \times \Z^d, \]
where $S_j = Q'_1 \times \cdots \times Q'_{j-1} \times Q_j \times Q'_{j+1} \times \cdots \times Q'_k$, and $Q'_i = Q_i \cup \{?\}$. We call $I \subset Q$ the set of \emph{initial states}, and $R \subset Q$ the set of \emph{rejecting states}.
\end{definition}

Note that all functions above are total.

\begin{definition}
Let $A = (Q, \Sigma, \delta, I, R)$ be a $k$PWDFA. An \emph{instantaneous description} or \emph{ID} of $A$ is an element of $\mathrm{ID}_A = (\Z^d)^k \times Q$. Given a configuration $x \in \Sigma^{\Z^d}$, we define the \emph{update function} $A_x : \mathrm{ID}_A \to \mathrm{ID}_A$. Namely, given $c = (\vec v^1, \ldots, \vec v^k, q_1, \ldots, q_k) \in \mathrm{ID}_A$, we define $A_x(c)$ as follows. If $(q_1, \ldots, q_k) \in R$ and $\vec v^1 = \cdots = \vec v^k$, then we say $c$ is \emph{rejecting}, and $A_x(c) = c$. Otherwise, $A_x(c) = (\vec w^1, \ldots, \vec w^k, p_1, \ldots, p_k)$, where $\vec w^j = \vec v^j + \vec u^j$ and
\[ \delta_j(q'_1, \ldots, q'_{j-1}, q_j, q'_{j+1}, \ldots, q'_k, x_{\vec v^j}) = (p_j, \vec u^j), \]
where we write $q'_i = q_i$ if $\vec v^i = \vec v^j$, and $q'_i = {?}$ otherwise. The \emph{run of $A$ on $x \in \Sigma^{\Z^d}$ from $c \in \mathrm{ID}_A$} is the infinite sequence $A_x^\infty(c) = (A_x^n(c))_{n \in \N}$. We say the run is \emph{accepting} if no $A_x^n(c)$ is rejecting. We define the \emph{subshift of $A$} by
\begin{align*}
S(A) = \{ x \in \Sigma^{\Z^d} \;|\; \forall q = (q_1, \ldots, q_k) \in I, \vec v \in \Z^d: A_x^\infty(\vec v, \ldots, \vec v, q) \mbox{~is accepting.} \}
\end{align*}
\end{definition}

We now define our hierarchy of interest:

\begin{definition}
We refer to the class of all $d$-dimensional SFTs (sofc shifts) over the alphabet $\Sigma = \{0, 1\}$ as simply $\SFT^d$ ($\sofic^d$, respectively). For $k > 0$, define
\[ S_k^d = \{ S(A) \;|\; A \mbox{~is a $d$-dimensional $k$PWDFA}. \} \]
\end{definition}

It is easy to see that $S_k^d \subset S_{k+1}^d$ for all $k > 0$, and that every $S_k^d$ only contains $\PI^0_1$ subshifts. Since a deterministic finite state automaton can clearly check any local property, we also have $\SFT^d \subset S_1^d$.

\begin{remark}
We note some robustness properties. While the definition only allows information sharing when several heads lie in the same cell, we may assume that heads can communicate if they are at most $t$ cells away from each other. Namely, if we had a stronger $k$-head automaton where such behavior is allowed, then we could simulate its computation step by $\Theta(k t^d)$ steps of a $k$PWDFA where the heads visit, one by one, the $\Theta(t^d)$ cells at most $t$ steps away from them, and remember which other heads they saw in which states. Also, while we allow the machines to move by any finite vector, we may assume these vectors all have length $0$ or $1$ by simulating a step of length $r$ by $r$ steps of length $1$. Finally, the classes $S_k^d$ are closed under conjugacy, rotation, mirroring and intersection.
\end{remark}

To compare these classes, we need to define a few subshifts and classes of subshifts. In most of our examples, the configurations contain the symbol $0$ in all but a bounded number of coordinates.

\begin{definition}
The \emph{$d$-dimensional $n$-sunny side up subshift} is the $d$-dimensional subshift $X^d_n \subset \{0, 1\}^{\Z^d}$ with forbidden patterns $\{ P \;|\; |P|_1 > n \}$. A $d$-dimensional subshift is \emph{$n$-sparse} if it is a subshift of $X^d_n$, and \emph{sparse} if it is $n$-sparse for some $n \in \N$. If $X$ is a $d_1$-dimensional subshift and $d_2 > d_1$, we define $X^{\Z^{d_2 - d_1}}$ as the $d_2$-dimensional subshift where the contents of every $d_1$-dimensional hyperplane $\{ \sum_{i=1}^{d_1} n_i \vec e_i \;|\; \vec n \in \Z^{d_1} \} \subset \Z^{d_2}$ are independently taken from $X$.
\end{definition}

An $n$-sparse subshift is one where at most $n$ symbols $1$ may occur, and the sunny side up subshifts are the ones with no additional constraints. The name sunny side up subshift is from \cite{PaSc10}. We called the $n$-sunny side up subshift \emph{the} $n$-sparse subshift in \cite{SaTo13b}, but feel that the terminology used here is a bit better.

We also use the following variation of the well-known mirror subshift.

\begin{definition}
\label{def:Mirror}
The \emph{$d$-dimensional mirror subshift} $X_\mathrm{mirror}^d \subset \{0, 1\}^{\Z^d}$ is defined by the following forbidden patterns.
\begin{itemize} 
\item All patterns $P$ of domain $\{ 0 \} \times \{ 0, 1, 2 \}^{d-1}$ such that the all-$1$ pattern of domain $\{ \vec 0, \vec e_i \}$ for some $i \in \{2, \ldots d\}$ occurs in $P$, but $|P|_0 \neq 0$.
\item All patterns $P$ of domain $\{ 0, k \} \times \{ 0, 1 \}^{d-1}$ for some $k > 1$ with $|P|_0 = 0$.
\item All patterns $P$ of domain $ \{ -k, k \} \times \{ 0 \}^{d-1} \cup \{ 0 \} \times \{ 0, 1 \}^{d-1}$ for some $k > 1$ where $P|_{\{ 0 \} \times \{ 0, 1 \}^{d-1}}$ contains no symbols $0$ and $P_{(-k, 0, \ldots, 0)} \neq P_{(k, 0, \ldots, 0)}$.
\end{itemize}
\end{definition}

Intuitively, the rules are that if two symbols $1$ are adjacent on some $(d-1)$-dimensional hyperplane perpendicular to $\vec e_1$, then that hyperplane must be filled with $1$'s, and there is at most one such hyperplane, whose two sides are mirror images of each other. In two dimensions, the hyperplane is just a vertical line.

Finally, we define a type of counter machine, which we will simulate by $2$- and $3$-head automata in the proofs of Proposition~\ref{prop:SparseCoRE} and Theorem~\ref{thm:ThreeIsAll}. This is essentially the model MP$1$RM (More Powerful One-Register Machine) defined in \cite{Sc72}. We could also use any other Turing complete machine with a single counter which supports multiplication and division, such as John Conway's FRACTRAN \cite{Co87}.

\begin{definition}
An \emph{arithmetical program} is a sequence of commands of the form
\begin{itemize}
\item Multiply/divide/increment/decrement $C$ by $m$,
\item If $(C \bmod m) = j$, goto $k$,
\item If $C = m$, goto $k$,
\item Halt,
\end{itemize}
where $j, m \in \N$ are arbitrary constants and $k \in \N$ refers to one of the commands.
\end{definition}

To run such a program on an input $n \in \N$, we initialize a single counter $C$ to $n$, and start executing the commands in order. The arithmetical commands work in the obvious way. We may assume the program never divides by a number unless it has checked that the value in $C$ is divisible by it, and never subtracts $m$ unless the value in $C$ is at least $m$. Thus, $C$ always contains a natural number. In the goto-statements, execution continues at command number $k$. The halt command ends the execution, and signifies that the program accepts $n$. It is well-known that this model is Turing complete; more precisely, we have the following.

\begin{lemma}[\cite{Sc72}]
\label{lem:ExpRE}
If a set $L \subset \N$ is recursively enumerable, then $\{2^n \;|\; n \in L\}$ is accepted by some arithmetical program.
\end{lemma}

\section{Results}

Our first results place the class $S_1^d$ between $\SFT^d$ and $\sofic^d$.

\begin{lemma}
\label{lem:OneGeZero}
In all dimensions $d$, we have $(X^1_1)^{\Z^{d-1}} \in S_1^d \setminus \SFT^d$.
\end{lemma}

\begin{proof}
Note that $X = (X^1_1)^{\Z^{d-1}}$ is the $d$-dimensional subshift where no row may contain two symbols $1$. First, we show $X$ is not an SFT: Suppose on the contrary that it is defined by a finite set of forbidden patterns with domain $[0, n-1]^d$ for some $n \in \N$. Consider the configurations $x^0, x^1 \in \Sigma^{\Z^d}$ where $x^i_{(0, 0)} = x^i_{(n, i)} = 1$ and $x^i_{\vec v} = 0$ for $\vec v \in \Z^d - \{(0, 0), (n, i)\}$. Since any pattern with domain $[0,n-1]^d$ occurs in $x^0$ if and only if it occurs in $x^1$, we have $x^0 \in X$ if and only if $x^1 \in X$, a contradiction since clearly $x^0 \notin X$ and $x^1 \in X$.

To show that $X \in 1\mathrm{PWDFA}$, we construct a one-head automaton for $X$. The idea is that the head will walk in the direction of the first coordinate, and increment a counter when it sees a symbol $1$. If the counter reaches $2$, the automaton rejects. More precisely, the automaton is $A_1 = (\{q_0,q_1,q_2\}, \{0, 1\}, \delta, \{q_0\}, \{q_2\})$, where $\delta(q_0, a) = (q_a, \vec e_1)$, $\delta(q_1, a) = (q_{1+a}, \vec e_1)$ and $\delta(q_2, a) = (q_2, \vec 0)$ for $a \in \{0, 1\}$. If there are two $1$'s on any of the rows of a configuration $x \in \Sigma^{\Z^d}$, say $x_{\vec v} = x_{\vec w} = 1$ where $\vec w = \vec v + n \vec e_1$ for some $n \geq 1$, then the run of $A_1$ on $x$ from $(q_0, \vec v)$ is not accepting, as the rejecting ID $(q_2,\vec w + \vec e_1)$ is entered after $n+1$ steps. Thus, $x \notin S(A_1)$. On the other hand, it is easy to see the if no row of $x \in \Sigma^{\Z^2}$ contains two symbols $1$, then $x \in S(A_1)$. \qed
\end{proof}

\begin{theorem}
\label{thm:Between}
In all dimensions $d$, we have $S_1^d \subset \sofic^d$, with equality if $d = 1$.
\end{theorem}

\begin{proof}
We first show $S_1^d \subset \sofic^d$. The proof of this is quite standard, see for example \cite{KaMo01}. Suppose $X \in S_1^d$, and let $A = (Q, \Sigma, \delta, I, R)$ be a $1$PWDFA accepting $X$. We construct an SFT $Y$ over the alphabet $2^Q \times \Sigma$, such that the second component of $Y$ contains exactly $X$. The forbidden patterns of $Y$ are
\begin{itemize}
\item every symbol $(Q', c) \in 2^Q \times \Sigma$ such that $I \not\subset Q'$ or $R \cap Q' \neq \emptyset$, and
\item every pair $\{\vec 0 \mapsto (Q_1, c_1), \vec v \mapsto (Q_2, c_2)\}$ such that $\delta(q_1, c_1) = (q_2, \vec v)$ for some $q_1 \in Q_1$ and $q_2 \notin Q_2$.
\end{itemize}
Now, if we initialize $A$ on the first component of some $y \in Y$, it is easy to see by induction that if it lies at $\vec v$ in state $q \in Q$ after some $n$ steps, then the first component of $y_{\vec v}$ contains $q$. Conversely, if $A$ accepts a configuration $x \in \Sigma^{\Z^d}$, then we collect the states of its infinite runs for every coordinate, and form a configuration $z \in (2^Q)^{\Z^d}$ with $(x,z) \in Y$.

It is well-known that a one-dimensional subshift is sofic if and only if it can be defined by a regular language of forbidden words \cite{LiMa95}. Since $2$-way deterministic finite automata only recognize regular languages, we have $\sofic^1 \subset S_1^1$, and the classes coincide. \qed
\end{proof}

\begin{remark}
For all dimensions $d_1 < d_2$, all $k$, and all subshifts $X \in S_k^{d_1}$, we have $X^{\Z^{d_2 - d_1}} \in S_k^{d_2}$, since a $d_2$-dimensional $k$PWDFA can simply simulate a $d_1$-dimensional one on any $d_1$-dimensional hyperplane. In particular, if $X \subset \Sigma^\Z$ is sofic, then $X^{\Z^{d-1}} \in S_1^d$ for any dimension $d$.
\end{remark}

Of course, since multidimensional SFTs may contain very complicated configurations, the same is true for the classes $S_1^d$. In particular, for all $d \geq 2$ there are subshifts in $S_1^d$ whose languages are co-RE-complete. However, just like in the case of SFTs, the sparse parts of subshifts in $S_1^d$ are simpler.

\begin{theorem}
\label{thm:1SparseRecursive}
Let the dimension $d$ be arbitrary, and let $X \in S_1^d$. For all $k$, the intersection $X \cap X_k^d$ is recursive.
\end{theorem}

\begin{proof}
Let $X = S(A)$ for a $1$PWDFA $A = (Q, \Sigma, \delta, I, R)$ that only takes steps of length $0$ and $1$. First, we claim that it is decidable whether a given configuration $y$ with at most $k$ symbols $1$ is in $Y$. We need to check whether there exists $\vec v \in \Z^d$ such that started from $\vec v$ in one of the initial states, $A$ eventually rejects $y$.

To decide this, note first that if $A$ does not see any symbols $1$, then it does not reject -- otherwise, the all-$0$ configuration would not be in $Y$. Define $W = \{ \vec v \in \Z^d \;|\; \|\vec v\| \leq |Q| \}$, and denote $\Z W = \{ n \vec w \;|\; n \in \Z, w \in W \}$.
Let $E \subset \Z^d$ be the convex hull of $D = \{ \vec v \in \Z^d \;|\; y_{\vec v} = 1 \}$, and let $F = E + W + W$. Note that no matter which initial state $A$ is started from, the only starting positions from which it can reach one of the symbols $1$ are those in
\[ W + \Z W + W + D \subset \Z W + F. \]
Namely, whenever $A$ takes $|Q|$ steps without encountering a symbol $1$, it must repeat a state. Thus, if $A$ is at least $2|Q|$ cells away from the nearest symbol $1$, then it must be ultimately periodically moving in some direction $\vec v \in \Z^d$ with $\|\vec v\| \leq |Q|$, repeating its state every $s \leq |Q|$ steps. If we denote by $(q_n, \vec v_n)_{n \leq N}$ the (finite or infinite) sequence of states and coordinates that $A$ visits before encountering a symbol $1$, then there are $a < b \leq |Q|$ with $q_a = q_b$. This implies that $\vec v_{a + k(b-a) + \ell} = \vec v_a + k (\vec v_b - \vec v_a) + \vec w_\ell$ for all $k \in \N$ and $\ell \leq b-a$ for which the coordinate is defined, where $\|\vec v_a - \vec v_0\|, \|\vec w_\ell\| \leq |Q|$. The claim follows, since $A$ must enter the domain $D$ in order to encounter a $1$.

Next, we show that we only need to analyze the starting positions in $G = W + W + W + F$. Namely, if $A$ enters the set $F$ for the first time after $a + k(b-a) + \ell$ steps and $k > 2 |Q| / \|\vec v_b - \vec v_a\|$, then the distance of the coordinate $\vec v_n$ from $F$ is at least $|Q|$ for all $n \leq a$. This means that if we initialize $A$ at the coordinate $\vec v_0 + \vec v_b - \vec v_a$ in the same state $q_0 \in Q$, then it will also enter $F$ for the first time in the state $q_{a + k(b-a) + \ell}$ and at the coordinate $\vec v_{a + k(b-a) + \ell}$.

From each starting position in the finite set $G$ and each initial state, we now simulate the machine until it first enters $F$ or exits $W + G$ (in which case it never enters $F$). Now, we note that if the machine re-exits $F$ after the first time it is entered, then it does not reject $y$. Namely, $F = E + W + W$ is convex and contains a $0$-filled border thick enough that $A$ must be in an infinite loop, heading off to infinity. Thus, if $A$ ever rejects $y$, it must do so by entering $F$ from $G$ without exiting $W + G$, then staying inside $F$, and rejecting before entering a loop, which we can easily detect. This finishes the proof of decidability of $y \in Y$.

Now, given a pattern $P$ with domain $D \subset \Z^d$, we need to decide whether it occurs in a configuration of $Y$. If $|P|_1 > k$, the answer is of course `no' since $Y$ is $k$-sparse, so suppose $|P|_1 \leq k$. Construct the configuration $y$ with $y|_D = P$ and $y_{\vec v} = 0$ for $\vec v \in \Z^d \setminus D$. If $y \in Y$, which is decidable by the above argument, then we answer `yes'. If $y \notin Y$ and $|P|_1 = k$, then we can safely answer `no'.

If $y \notin Y$ and $|P|_1 < k$, then we have found a rejecting run of $A$ that only visits some finite set of coordinates $C \subset \Z^d$. If there exists $x \in Y$ such that $x|_D = P$, then necessarily $x_{\vec v} = 1$ for some $\vec v \in C \setminus D$. For all such $\vec v$, we construct a new pattern by adding $\{ \vec v \mapsto 1 \}$ into $P$, and call this algorithm recursively on it. If one of the recursive calls returns `yes', then we answer `yes' as well. Otherwise, we answer `no'. The correctness of this algorithm now follows by induction. \qed
\end{proof}

For the previous result to be nontrivial, it is important to explicitly take the intersection with a sparse subshift instead of assuming that $X$ is sparse, for the following reason.

\begin{proposition}
\label{prop:OneHeadNoSparse}
For all dimensions $d \geq 2$, the class $S_1^d$ contains no nontrivial sparse subshifts.
\end{proposition}

\begin{proof}
Let $A$ be a $1$PWDFA such that $S(A)$ is sparse and contains at least two configurations. We may assume that $X_1^d \subset S(A)$ by recoding if necessary. Recall the notation of the proof of Theorem~\ref{thm:1SparseRecursive}. It was shown there that if $A$ can reach a position $\vec v \in \Z^d$ from the origin without encountering a $1$, then $\vec v \in W + W + \Z W$. Let $V \subset \Z^d$ be an infinite set such that $\vec v - \vec w \notin \Z W + W + W$ for all $\vec v \neq \vec w \in V$. One exists since $d \geq 2$. Define $x \in \Sigma^{\Z^d}$ by $x_{\vec v} = 1$ if and only if $\vec v \in V$. Then $A$ accepts $x$, since it encounters at most one symbol $1$ on every run on $x$, contradicting the sparsity of $S(A)$. \qed
\end{proof}

Next, we show that two heads are already quite powerful in the one- and two-dimensional settings, and such results do not hold for them. In two dimensions, some type of searching is also possible with just two heads.

\begin{proposition}
\label{prop:Sunny}
The $k$-sunny side up shift $X_k^2$ is in $S_2^2$ for all $k$.
\end{proposition}

\begin{proof}
For $a, b, c, d \in \N$ with $a+b+c+d = k+1$, we construct a two-head automaton $A_{a, b, c, d}$ with the following property: when started on top of a symbol $1$ at the coordinate $\vec 0$, the automaton rejects a configuration if and only if
\begin{itemize}
\item the quarterplane $\N \times \N$ contains at least $a$ symbols $1$,
\item the quarterplane $(-\infty, -1] \times \N$ contains at least $b$ symbols $1$,
\item the quarterplane $(-\infty, -1] \times (-\infty, -1]$ contains at least $c$ symbols $1$, and
\item the quarterplane $\N \times (-\infty, -1]$ contains at least $d$ symbols $1$.
\end{itemize}
Clearly, the intersection of the subshifts accepted by the finitely many automata $A_{a, b, c, d}$ is precisely $S_k$.

Since the four cases are essentially symmetric, it is enough to construct an automaton $A_a$ that checks that there are at least $a$ symbols $1$ on the top right quarterplane, and then returns to its starting position. First, the automaton checks that it is indeed on top of a symbol $1$, and enters an infinite loop if not.

The two heads of $A_a$ are called the \emph{L-head} and the \emph{diagonal head}. Both heads remember a number $j \in [0, a]$, the number of the diagonal head being called the \emph{count}, and the other the \emph{height}. In the initial state, the count is $1$ and the height is $0$. We inductively preserve the following invariant: If the two heads are at $(0, n)$ and the count is $j < a$, then there are exactly $j$ symbols $1$ in the coordinates $D = \{(m, m') \;|\; m, m' \in \N, m+m' \leq n\}$, and if $j = a$, then $D$ contains at least $a$ symbols $1$; the height is precisely the number of $1$s on the column between $(0, n)$ and $(0, 0)$. We explain how, if the automaton is in coordinate $(0, n)$ with count $j$ and height $h$ so that the invariant holds, it can move to the coordinate $(0,n+1)$, preserving the invariant.

The automaton sends its L-head down at speed $1$, and the diagonal head southeast at speed $1/2$ (that is, the diagonal head moves every second step). When the L-head finds the coordinate $(0,0)$ (which it can determine based on the height), it turns right, again using the height counter to remember the number of $1$'s it has seen on the row. The two heads meet at $(n, 0)$. Now, the heads move one step to the right, possibly updating the width counter. The heads then repeat the procedure in reverse, with the difference that the diagonal head increments the count value for every $1$ it encounters on its way northwest, up to the value of $a$. The heads meet at $(0, n+1)$, and the invariant is preserved.

Finally, if the count is $a$ and the heads are at a position $(0,n)$, they can return to the origin together with the aid of the height counter. \qed
\end{proof}

The following proposition gives the separation of the classes $S_1^d$ and $S_2^d$ for $d \leq 2$. It can be thought of as an analogue of the well-known result that two counters are enough for arbitrarily complicated (though not arbitrary) computation.

\begin{proposition}
\label{prop:SparseCoRE}
For $d \leq 2$, there is a $2$-sparse co-RE-complete subshift $X \in S_2^d$.
\end{proposition}

\begin{proof}
We only prove the case $d = 2$, as the one-dimensional case is even easier. Let $X$ be the subshift of $X_2^2$ where either the two symbols $1$ are on different rows, or their distance is not $2^n$ for any $n \in L$, for a fixed RE-complete set $L \subset \N$.

To prove $X \in S_2^2$, we construct a $2$PWDFA $A$ for it. The heads of $A$ are called the `zig-zag head' and the `counter head'. Since $S_2^2$ is closed under intersection, Proposition~\ref{prop:Sunny} shows that we may restrict our attention to configurations of $X_2^2$. First, our machine checks that it is started on a symbol $1$ and another symbol $1$ occurs on the same row to the left, by doing a left-and-right sweep with the zig-zag head. Otherwise, $A$ runs forever without halting. The rightmost $1$ is ignored during the rest of the computation, and from now on, we refer to the leftmost $1$ as the \emph{pointer}. Since the heads never leave the row on which they started, they can keep track of whether they are to the right or to the left of the rightmost $1$.

We think of the distance of the counter head from the pointer as the value of a counter $C$ of an arithmetic program accepting the language $L' = \{2^n \;|\; n \in L\}$ (which exists by Lemma~\ref{lem:ExpRE}). We simulate this program using the two heads as follows: The finite state of the zig-zag head will store the state of the program. If the counter of the arithmetical program contains the value $C$ and the pointer is at $\vec v \in \Z^2$, then both heads are at $\vec v + (C,0)$ (except for intermediate steps when a command of the program is being executed). See Figure~\ref{fig:Computing}. To increment or decrement $C$ by $m$, the zig-zag head and the counter head simply move $m$ steps to the left or right, staying together. To check $C = m$, the zig-zag head moves $m$ steps to the left and looks for the pointer, and to check $(C \bmod m) = j$, the zig-zag head makes a left-and-right sweep, visiting the pointer and returning to the counter head, using its finitely many states to compute the remainder.

Multiplications and divisions are done by standard signal constructions. For example, to move the zig-zag head and the counter head from $\vec v + (C, 0)$ to $\vec v + (C/2, 0)$ (assuming it has been checked that $C$ is even), the counter head starts moving left at speed $1$, and the zig-zag head at speed $3$, bouncing back from the pointer, and the two meet at exactly $\vec v + (C/2, 0)$. It is easy to construct such pairs of speeds for multiplication or division by any fixed natural number.

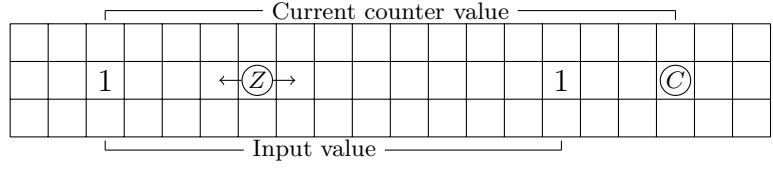
\begin{figure}
\begin{center}
\begin{tikzpicture}[scale=.5]

\node[draw,circle,inner sep=0,minimum size=.4cm] (zigzag) at (6.5,1.5) {$Z$};
\node[draw,circle,inner sep=0,minimum size=.4cm] (counter) at (17.5,1.5) {$C$};
\draw[->] (zigzag) edge (5.5,1.5);
\draw[->] (zigzag) edge (7.5,1.5);

\node (pointer) at (2.5,1.5) {\large $1$};
\node () at (14.5,1.5) {\large $1$};

\node (text) at (10,3.35) {Current counter value};
\draw (text) edge (2.5,3.35) edge (17.5,3.35);
\draw (2.5,3.35) -- ++(0,-.25);
\draw (17.5,3.35) -- ++(0,-.25);

\node (text2) at (8,-.35) {Input value};
\draw (text2) edge (2.5,-.35) edge (14.5,-.35);
\draw (2.5,-.35) -- ++(0,.25);
\draw (14.5,-.35) -- ++(0,.25);

\draw (0,0) grid (20,3);

\end{tikzpicture}
\end{center}
\caption{Simulating an arithmetical program with two heads, labeled $Z$ for zig-zag and $C$ for counter. The leftmost $1$ is the pointer, and empty squares contain $0$-symbols.}
\label{fig:Computing}
\end{figure}

If the arithmetical program eventually halts, then $A$ rejects the configuration, and otherwise it simulates the program forever. Now, let $x \in X_2^2$ be arbitrary. If $A$ is not started on the rightmost $1$ of a row of $x$ that contains two $1$'s, then it does not reject $x$. Suppose then that this holds and let $\ell \in \N$ be the distance between the two $1$'s, so that $A$ starts simulating the arithmetical program as described above, with input value $\ell$. If $\ell \in L'$, then the program eventually halts and the automaton rejects, and we have $x \notin X$. Otherwise, the program and thus the automaton run forever, and $x \in X$ since $A$ does not reject $x$ from any starting position. This shows that $S(A) = X$. \qed
\end{proof}

We do not believe that \emph{all} $2$-sparse co-RE-complete subshifts are in $S_2^d$ for $d \leq 2$, but we cannot prove this. In three or more dimensions, however, we obtain the following analogue of Proposition~\ref{prop:OneHeadNoSparse}, which is proved similarly.

\begin{theorem}
\label{thm:ThreeIsMore}
For all dimensions $d \geq 3$, the class $S_2^d$ contains no nontrivial sparse subshifts.
\end{theorem}

\begin{proof}
Let $A$ be a $2$PWDFA taking only steps of length $0$ or $1$ such that $S(A)$ is sparse and contains at least two configurations. We may again assume that $X_1^d \subset S(A)$. As in the proof of Theorem~\ref{thm:1SparseRecursive}, it is easy to see that there exists some $p \in \N$ such that, denoting $W = \{ \vec v \in \Z^d \;|\; \|\vec v\| \leq p \}$ and $\Z W = \{ n \vec w \;|\; n \in \Z, \vec w \in W \}$, we have the following. Let the two heads of $A$ be initialized on some coordinates $\vec v = \vec v_0 \in \Z^2$ and $\vec w = \vec w_0 \in \Z^2$ in any states, and denote by $(\vec v_n)_{n \leq N}$ and $(\vec w_n)_{n \leq N}$ their itineraries up to some timestep $N \in \N$. If we have $\| \vec v - \vec w \| \leq p$ ($\| \vec v - \vec w \| > p$), then $\vec v_n \in \vec v + \Z W + W$ and $\vec v_n \in \vec v + \Z W + W$ until either head sees a symbol $1$ (either head sees a symbol $1$ or the heads meet each other, respectively). In the former case, note that the heads may travel together, so that their `combined state' can have a period greater than $|Q|$.

Analogously to the proof of Proposition~\ref{prop:OneHeadNoSparse}, let $V \subset \Z^d$ be an infinite set such that $\vec v - \vec w \notin \Z W + \Z W + W + W$ for all $\vec v, \vec w \in V$. Define $x \in \Sigma^{\Z^d}$ by $x_{\vec v} = 1$ if and only if $\vec v \in V$. We prove that $x$ is accepted by $A$, contradicting the sparsity of $S(A)$. We may assume that $A$ is started at some position $\vec w \in \Z^d$ and encounters a $1$ at the origin after some number of steps.

By the first paragraph, both heads stay in the region $\vec w + \Z W + W$ until the origin is found, say by the first head. Then $\vec w \in \Z W + W$, so the second head stays in the domain $\Z W + \Z W + W + W$ until it encounters the origin or the first head. The first head is restricted to the domain $\Z W + W$ until it meets the second head, so the heads cannot reach any coordinate $\vec v \in V \setminus \{\vec 0\}$ before this. But if the heads meet, they must do so in a coordinate of $\Z W + W$, and after this, they are confined to the domain $\Z W + \Z W + W + W$ until one of them reaches the origin again. Thus, the heads never reach a symbol $1$ other than the origin, and since $X_1^d \subset S(A)$, the configuration $x$ must be accepted.
\qed
\end{proof}

There are no nontrivial restrictions for sparse sofic shifts.

\begin{theorem}
\label{thm:SparseIsSofic}
For all dimensions $d \geq 2$, every sparse co-RE subshift is in $\sofic^d$.
\end{theorem}

\begin{proof}
We show the result in two dimensions, the general case is similar.

Here, we consider a larger alphabet than $\Sigma = \{0, 1\}$. Namely, we will show that every $\PI^0_1$ subshift $X$ over $\{0, \ldots, k\}$ containing all symbols except $0$ at most once is sofic. This proves the original claim, since sofic shifts are closed under renaming the symbols. Let $T$ be a Turing machine enumerating a sequence $(P_i)_{i \in \N}$ of forbidden patterns for $X$. We will construct an SFT $Y \subset \{0, \ldots, k\}^{\Z^2} \times Z$, where $Z$ is also an SFT, such that the projection of $Y$ to the first layer is exactly $X$. The SFT $Z$ also has several layers, and its alphabet is $\{\$, 0, 1, 2\}^k \times ((Q \times \Gamma) \cup \Gamma \cup \{\#\})^k$, where $Q$ and $\Gamma$ are the state set and tape alphabet of another Turing machine $T'$ to be described later, respectively, with $\{\$, 0, 1, 2\}^k \subset \Gamma$. We denote by $Y_i$ ($Z_i$) the projection of $Z$ onto the $i$'th layer of the first (second, respectively) component of the product. The $Y_i$ are called \emph{signal layers} and the $Z_i$ \emph{computation layers}.

For each signal layer $Y_i \subset \{\$, 0, 1, 2\}^{\Z^2}$, the $\$$-symbols correspond exactly to the $i$-symbols in the first layer of $Y$, in the sense that for a configuration $y = (x, y^1, \ldots, y^k, z^1, \ldots, z^k) \in Y$ and $\vec v \in \Z^2$, we have $y^i_{\vec v} = \$$ if and only if $x_{\vec v} = i$. The forbidden patterns of each $Y_i$ are exactly the $2 \times 2$-patterns that do not occur in the pattern
\[ \begin{array}{ccccc}
	0 & 0 &  1 & 1 & 1 \\
	0 & 0 &  1 & 1 & 1 \\
	0 & 0 & \$ & 1 & 1 \\
	0 & 0 &  2 & 2 & 1 \\
	0 & 0 &  2 & 2 & 2 \\
\end{array} \]
It is easy to see that $Y_i$ contains at most one occurrence of $\$$, and thus the first layer of $Y$ contains at most one occurrence of $i$.

We now define the computation layers $Z_i$. First, every L-shaped pattern $\begin{smallmatrix} a & \\ b & c \end{smallmatrix}$ where $\#$ occurs is forbidden, except if it satisfies either $a = b = c = \#$ or $a \neq \# = b = c$. In the latter case, we require that the $\Gamma$-component of $a$ is exactly the corresponding symbol $a' \in \{\$, 0, 1, 2\}^k$ on the product layer $\prod_{i=1}^k Y_i$, that $a$ has a $Q$-component if and only if $a'_i = \$$, and that the $Q$-component is then the initial state of $T'$. In particular, in a configuration $y \in Y$ whose first layer contains the symbol $i$, the $Z_i$-layer $z_i$ contains a downward half plane of $\#$, on top of which is a horizontal row of the product layer $\prod_{i=1}^k Y_i$, and one read-write head of the Turing machine $T'$ in its initial state. Using further $2 \times 2$ forbidden patterns, we require that on the subsequent rows of $z_i$, a computation of $T'$ is simulated, and a halting state results in a tiling error.

Now, let $i \in \{1, \ldots, k\}$ be such that $i$ occurs in $x$ at a position $\vec v \in \Z^2$ whose y-coordinate is minimal. On the layer $Z_i$, for any given $n \in \N$, the simulated machine $T'$ can compute the square pattern $S_n = x|_{[-n,n]^2 + \vec v}$ of the first layer of $y$, since it can infer the relative positions of all symbols $j \in \{1, \ldots, k\}$ from its initial data. See Figure~\ref{fig:TM} for a visualization. Now, we define $T'$ so that for all $n \in \N$ in turn, it computes the aforementioned pattern $S_n$ and the first $n$ patterns $(P_j)_{j = 0}^{n-1}$ given by $T$, checks whether some $P_j$ occurs in $S_n$, and halts if this holds.

Now, a given $x \in \{0, \ldots, k\}^{\Z^d}$ is a projection of a configuration of $Y$ if and only if every symbol $i \in \{1, \ldots, k\}$ occurs in $x$ at most once, and for the one occurring at $\vec v \in \Z^2$ as above, no $P_j$ for $j < n$ occurs in $x|_{[-n,n] + \vec v}$ for any $n \in \N$. This is equivalent to $x \in X$. \qed
\end{proof}

\begin{figure}
\begin{center}
\begin{tikzpicture}[scale=.5]

\fill[black!20] (0,0) rectangle (14,3);

\foreach \x in {0,...,13}{
	\foreach \y in {0,...,2}{
		\node () at (\x+.5,\y+.5) {\#};
	}
}

\draw[thick, fill=black!20] (4.5,3.5) circle (.4);
\draw[very thick] (4,0) -- (4,3) -- ++(2,0) -- ++(0,-1) -- ++(1,0) -- ++(0,-1) -- ++(1,0) -- ++(0,-1);

\draw[thick, fill=black!20] (9.5,5.5) circle (.4);
\draw[very thick] (9,0) -- (9,5) -- ++(2,0) -- ++(0,-1) -- ++(1,0) -- ++(0,-1) -- ++(1,0) -- ++(0,-1) -- ++(1,0);

\foreach \x/\hy in {0/7,1/7,2/7,3/6}{
	\foreach \y in {3,...,\hy}{
		\node () at (\x+.5,\y+.5) {00};
	}
}
\foreach \x/\hy in {5/3,6/7,7/7,8/7}{
	\foreach \y in {3,...,\hy}{
		\node () at (\x+.5,\y+.5) {10};
	}
}
\foreach \x in {9,10,11}{
	\foreach \y in {3,...,7}{
		\node () at (\x+.5,\y+.5) {12};
	}
}
\foreach \x in {12,13}{
	\foreach \y in {3,...,7}{
		\node () at (\x+.5,\y+.5) {11};
	}
}

\node () at (4.5,3.5) {$q_0$};
\node () at (5.5,4.5) {$q_1$};
\node () at (5.5,5.5) {$q_2$};
\node () at (4.5,6.5) {$q_3$};
\node () at (3.5,7.5) {$q_4$};
\node () at (4.5,4.5) {$\gamma_0$};
\node () at (4.5,5.5) {$\gamma_0$};
\node () at (5.5,6.5) {$\gamma_1$};
\node () at (5.5,7.5) {$\gamma_1$};
\node () at (4.5,7.5) {$\gamma_2$};

\draw (0,0) grid (14,8);

\end{tikzpicture}
\end{center}
\caption{Simulating a Turing Machine on the computation layer $Z_1$, with $k = 2$. The two signal layers $Y_1$ and $Y_2$ are also shown, with the filled circles representing the $\$$-symbols, and the bordered areas containing $2$-symbols. Note that the grid squares show the tape of $T'$, not the contents of the signal layers, and that the latter can be inferred from the former. The $q_i$ are states of $T'$, and the $\gamma_i$ are its tape symbols.}
\label{fig:TM}
\end{figure}

Combining Theorem~\ref{thm:SparseIsSofic}, Theorem~\ref{thm:ThreeIsMore} and Proposition~\ref{prop:OneHeadNoSparse}, we obtain the following.

\begin{corollary}
\label{cor:Between}
For all dimensions $d \geq 2$, we have $S_1^d \subsetneq \sofic^d$, and for all dimensions $d \geq 3$, we have $S_2^d \not\subset \sofic^d$.
\end{corollary}

While Theorem~\ref{thm:SparseIsSofic} shows that all sparse $S_2^d$ subshifts are sofic, we can show that this is not true in general. In particular, the next result shows that $S_1^d$ is properly contained in $S_2^d$ for all $d \geq 2$.

\begin{proposition}
\label{prop:Mirror}
In all dimensions $d \geq 2$, we have $X_{\mathrm{mirror}}^d \in S_2^d \setminus \sofic^d$.
\end{proposition}

\begin{proof}
The proof of $X_{\mathrm{mirror}}^d \notin \sofic^d$ is completely standard both in the theory of subshifts and in the theory of picture languages, although we do not have a direct reference for it. The same argument is applied in \cite[Example 2.4]{KaMa13} to a slightly different subshift.

To show that $X_{\mathrm{mirror}}^d \in S_2^d$, we describe a $2$PWDFA for it. Using the fact that $S_2^d$ is closed under intersection, we restrict to the SFT defined by the first point of Definition~\ref{def:Mirror}. We can also assume there is at most one hyperplane of symbols $1$, as this is checked by a $1$PWDFA that walks in the direction of the first axis from its initial position, and halts if it sees the pattern $\{ \vec 0 \mapsto 1, \vec e_2 \mapsto 1 \}$ twice.

Under these assumptions, the mirror property is easy to check. One of the heads memorizes the bit in the initial position in its finite memory. Then, one of the heads starts traveling to the direction $\vec e_1$, and the other to $\vec e_1 + \vec e_2$. If the latter sees a hyperplane of symbols $1$, it turns to the direction $\vec e_1 - \vec e_2$. If the heads meet, they check that the bit in the initial position matches the bit under the current position, and if not, the configuration is rejected. \qed
\end{proof}

Finally, we collapse the hierarchy. This can be thought of as an analogue of the well-known result that three counters are enough for \emph{all} computation.

\begin{theorem}
\label{thm:ThreeIsAll}
In all dimensions $d$, the classes $S_k^d$ for $k \geq 3$ coincide with the class of co-RE subshifts.
\end{theorem}

\begin{proof}
We only need to show that $S_3^d$ contains all $\PI^0_1$ subshifts. Namely, $S_k^d \subset S_{k+1}^d$ holds for all $k > 0$, and since a Turing machine can easily enumerate patterns supporting a rejecting computation of a multihead finite automaton, every $S_k^d$ subshift is also $\PI^0_1$.

Let $T$ be a Turing machine that, when started from the initial configuration $c_0$ with empty input, outputs a sequence $(P_i)_{i \in \N}$ of patterns by writing each of them in turn to a special output track, and visiting a special state $q_\mathrm{out}$. We construct a $3$PWDFA $A_T$ accepting exactly those configurations where no $P_i$ occurs. The heads of $A_T$ are called the \emph{pointer head}, the \emph{zig-zag head}, and the \emph{counter head}. The machine has a single initial state, and when started from any position $\vec v \in \Z^d$ of a configuration $x$, it checks that no $P_i$ occurs in $x$ at $\vec v$. Since $A_T$ is started from every position, it will then forbid all translates of the $P_i$.

The machine simulates an arithmetical program as in the proof of Proposition~\ref{prop:SparseCoRE}, but in place of the `leftmost symbol $1$', we use the pointer head. The crucial difference here is that unlike a symbol $1$, the pointer head can be moved freely. This allows us to walk around the configuration, and extract any information we want from it. The arithmetical program simulates Algorithm~\ref{alg:Seeker}, which finally simulates the Turing machine $T$.

\begin{algorithm}
\caption{The algorithm that the three-head automaton $A_T$ simulates.}
\label{alg:Seeker}
\begin{algorithmic}[1]

\State $c \gets c_0$ \Comment{A configuration of $T$, set to the initial configuration}
\State $\vec u \gets \vec 0 \in \Z^d$ \Comment{The position of the pointer head relative to the initial position}
\State $P : \emptyset \to \{0,1\}$ \Comment{A finite pattern at the initial position}

\Loop
	\Repeat
		\State $c \gets \textsc{NextConf}_T(c)$ \Comment{Simulate one step of $T$}
	\Until{$\textsc{State}(c) = q_{\mathrm{out}}$} \Comment{$T$ outputs something}
	\State $P' \gets \textsc{OutputOf}(c)$ \Comment{A forbidden pattern}
	\While{$D(P') \not\subset D(P)$}
		\State $\vec w \gets \textsc{LexMin}(D(P) \setminus D(P'))$ \Comment{The lexicographically minimal vector}
		\While{$\vec u \neq \vec w$}
			\State $\vec d \gets \textsc{NearestUnitVector}(\vec w - \vec u)$ \Comment{Nearest unit vector in $\Z^d$}
			\State $\textsc{MoveBy}(\vec d)$ \Comment{Move the heads of $A_T$ to the given direction}
			\State $\vec u \gets \vec u + \vec d$
		\EndWhile
		\State $b \gets \textsc{ReadSymbol}$ \Comment{Read the symbol of $x$ under the pointer head}
		\State $P \gets P \cup \{ \vec u \mapsto b \}$ \Comment{Expand $P$ by one coordinate}
	\EndWhile
	\If{$P|_{D(P')} = P'$} \textbf{halt} \Comment{The forbidden pattern $P'$ was found}
	\EndIf
\EndLoop

\end{algorithmic}
\end{algorithm}

The algorithm remembers a finite pattern $P = x|_{D(P) + \vec v}$, where $\vec v \in \Z^d$ is the initial position of the heads, and a vector $\vec u \in \Z^d$ containing $\vec w - \vec v$, where $\vec w \in \Z^d$ is the current position of the pointer. The machine $T$ is simulated step by step, and whenever it outputs a forbidden pattern $P'$, the algorithm checks whether $D(P)$ contains its domain. If so, it then checks whether $x|_{D(P') + \vec v} = P'$. If this holds, then the algorithm halts, the arithmetical program simulating it halts, and the automaton $A_T$ moves all of its heads to the pointer and rejects. If $P'$ does not occur, the simulation of $T$ continues.

If $D(P')$ is not contained in $D(P)$, then the algorithm expands $P$, which is done in the outer \textbf{while}-loop of Algorithm~\ref{alg:Seeker}. To find out the contents of $x$ at some coordinate $\vec w + \vec v$ for $\vec w \in D(P')$, the algorithm chooses a unit direction (one of $\pm \vec e_i$ for $i \in \{1, \ldots, d\}$) that would take the pointer head closer to $\vec w + \vec v$, and signals it to $A_T$ via the arithmetical program. In a single sweep of the zig-zag head to the pointer and back, $A_T$ can easily move all of its heads one step in any unit direction. Then the simulation continues, and the algorithm updates $\vec u$ accordingly. When $\vec u = \vec w$ finally holds, the algorithm orders $A_T$ to read the symbol $x_{\vec v + \vec u}$ under the pointer, which is again doable in a single sweep. The bit $b = x_{\vec v + \vec u}$ is given to the algorithm, which expands $P$ by defining $P_{\vec u} = x_{\vec v + \vec u}$.

For a configuration $x$ and initial coordinate $\vec v \in \Z^d$, the automaton $A_T$ thus computes the sequence of patterns $(P_i)_{i \in \N}$ and checks for each $i \in \N$ whether $x|_{D(P_i) + \vec v} = P_i$ holds, rejecting if it does. Since $\vec v$ is arbitrary, we have $x \in S(A_T)$ if and only if no $P_i$ occurs in $x$. Thus $S_3^d$ contains an arbitrary $\PI^0_1$ subshift. \qed
\end{proof}

The basic comparisons obtained above are summarized in Figure~\ref{fig:Comparison}.

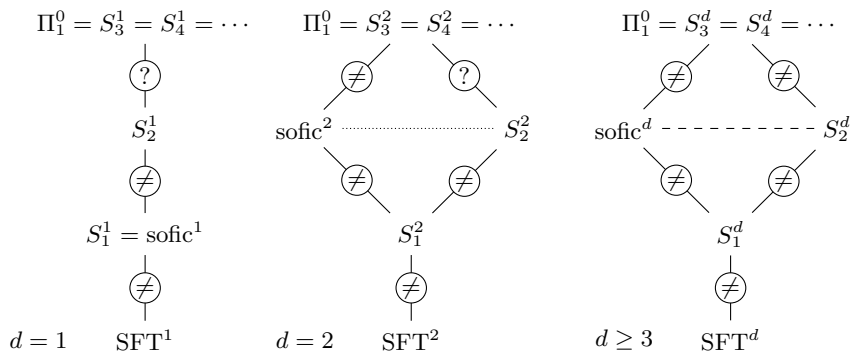
\begin{figure}
\begin{center}
\begin{tikzpicture}[scale = 1.4]

\begin{scope}
	\node () at (-1,0) {$d = 1$};
	\node (SFT) at (0,0) {$\SFT^1$};
	\node (S1) at (0, 1) {$S_1^1 = \sofic^1$};
	\node (S2) at (0, 2) {$S_2^1$};
	\node (S3) at (0, 3) {$\PI^0_1 = S_3^1 = S_4^1 = \cdots$};

	\draw 
		(SFT) edge node[draw,circle,fill=white,inner sep=0pt,minimum size=.4cm] {$\neq$} (S1)
		(S1)  edge node[draw,circle,fill=white,inner sep=0pt,minimum size=.4cm] {$\neq$} (S2)
		(S2) edge node[draw,circle,fill=white,inner sep=0pt,minimum size=.4cm] {?} (S3);
\end{scope}

\begin{scope}[xshift=2.5cm]
	\node () at (-1,0) {$d = 2$};
	\node (SFT) at (0,0) {$\SFT^2$};
	\node (S1) at (0, 1) {$S_1^2$};
	\node (sofic) at (-1, 2) {$\sofic^2$};
	\node (S2) at (1, 2) {$S_2^2$};
	\node (S3) at (0, 3) {$\PI^0_1 = S_3^2 = S_4^2 = \cdots$};

	\draw 
		(SFT) edge node[draw,circle,fill=white,inner sep=0pt,minimum size=.4cm] {$\neq$} (S1)
		(S1)  edge node[draw,circle,fill=white,inner sep=0pt,minimum size=.4cm] {$\neq$} (S2)
		edge node[draw,circle,fill=white,inner sep=0pt,minimum size=.4cm] {$\neq$} (sofic)
		(sofic) edge node[draw,circle,fill=white,inner sep=0pt,minimum size=.4cm] {$\neq$} (S3)
		(S2) edge node[draw,circle,fill=white,inner sep=0pt,minimum size=.4cm] {?} (S3);
	\draw[densely dotted] (sofic) edge (S2);
\end{scope}

\begin{scope}[xshift=5.5cm]
	\node () at (-1,0) {$d \geq 3$};
	\node (SFT) at (0,0) {$\SFT^d$};
	\node (S1) at (0, 1) {$S_1^d$};
	\node (sofic) at (-1, 2) {$\sofic^d$};
	\node (S2) at (1, 2) {$S_2^d$};
	\node (S3) at (0, 3) {$\PI^0_1 = S_3^d = S_4^d = \cdots$};

	\draw 
		(SFT) edge node[draw,circle,fill=white,inner sep=0pt,minimum size=.4cm] {$\neq$} (S1)
		(S1)  edge node[draw,circle,fill=white,inner sep=0pt,minimum size=.4cm] {$\neq$} (S2)
		edge node[draw,circle,fill=white,inner sep=0pt,minimum size=.4cm] {$\neq$} (sofic)
		(sofic) edge node[draw,circle,fill=white,inner sep=0pt,minimum size=.4cm] {$\neq$} (S3)
		(S2) edge node[draw,circle,fill=white,inner sep=0pt,minimum size=.4cm] {$\neq$} (S3);
	\draw[dashed] (sofic) edge (S2);
\end{scope}

\end{tikzpicture}
\end{center}
\caption{A comparison of our classes of subshifts. The solid, dashed and dotted lines denote inclusion, incomparability and an unknown relation, respectively, as we only know $S_2^d \not\subset \sofic^d$ for $d = 2$.}
\label{fig:Comparison}
\end{figure}

\section{The Classes $S_2^1$ and $S_2^2$}

A major missing link in our classification is the separation of $S_2^d$ and $S_3^d$ in dimensions $d \leq 2$. We leave this problem unsolved, but state the following conjecture.

\begin{conjecture}
\label{con:Sharper}
For $d \leq 2$, there exists a sparse co-RE subshift which is not in $S_2^d$. In particular we have $S_2^d \subsetneq S_3^d$, and $\sofic^2$ and $S_2^2$ are incomparable.
\end{conjecture}

Recall from the proof of Proposition~\ref{prop:SparseCoRE} that two counters are enough for a plane-walking automaton to simulate any arithmetical program in a sparse subshift. It is known that two-counter machines (which are basically equivalent to arithmetical programs by \cite{Sc72}) cannot compute all recursive functions, and in particular cannot recognize the set of prime numbers \cite{IbTr93}. A natural candidate for realizing Conjecture~\ref{con:Sharper} in the one-dimensional case would thus be the subshift $X \subset X_2^1$ where the distance of the two $1$'s cannot be a prime number.

However, instead of simply simulating an arithmetical program, the automaton may use the position of the rightmost $1$ in the middle of the computation, and a priori compute something an ordinary arithmetical program cannot. In some sense it thus simulates an arithmetical program that remembers its input. Conversely, we also believe that a run of a $2$PWDFA on a $2$-sparse subshift can be simulated by such a machine. All currently known proof techniques for limitations of two-counter machines break down if one is allowed to remember the input value, which raises the following question.

\begin{question}
Can arithmetical programs (or two-counter machines) that remember their input (for example, in the sense that they can check whether the current counter value is greater than the input) recognize all recursively enumerable sets? In particular, can they recognize the set of prime numbers?
\end{question}

Other tools for separating classes of multihead automata are diagonalization, where an automaton with much more than $k$ heads can analyze the behavior of one with $k$ heads, and choose to act differently from it on some inputs, and computability arguments, where algorithms of certain complexity can only be computed by machines with enough heads. Unfortunately, these approaches cannot separate $S_2^d$ from $S_3^d$, since both are capable of universal computation.

\bibliographystyle{plain}
\bibliography{../../../bib/bib}{}

\end{document}